\documentclass[11pt]{article}
\usepackage[utf8]{inputenc}
\usepackage{amsmath,amsthm,amssymb}
\usepackage{mathrsfs}
\usepackage{mathtools}
\usepackage[margin=1in]{geometry}
\usepackage{hyperref}

\newtheorem{theorem}{Theorem}
\newtheorem{lemma}[theorem]{Lemma}

\theoremstyle{definition}
\newtheorem{definition}[theorem]{Definition}
\theoremstyle{remark}
\newtheorem{remark}[theorem]{Remark}

\newcommand{\NP}{\mathscr{N\!P}}

\newcommand{\Vone}{V^0_1}
\newcommand{\Sat}{\mathrm{Sat}}
\newcommand{\WF}{\mathrm{WF}}
\newcommand{\WFH}{\mathrm{WFH}}
\newcommand{\NZ}{\mathrm{NZ}}
\newcommand{\HIT}{\mathrm{HIT}}
\newcommand{\Grd}{\mathrm{Grd}}
\newcommand{\Adm}{\mathrm{Adm}}
\newcommand{\HScert}{\mathrm{HS}^{\mathrm{c}}}
\newcommand{\dWPHPcert}{\mathrm{dWPHP}^{\mathrm{c}}}

\newcommand{\SigmaB}{\Sigma^B_0}
\newcommand{\code}[1]{\ulcorner #1 \urcorner}

\title{Unconditional $\Vone$-independence of a certified hitting-set principle}
\author{%
  Martin Kolář%
}
\date{}

\begin{document}
\maketitle

\begin{abstract}
We show that a certified formalization of the hitting-set-existence axiom of
Atserias and Tzameret, instantiated on the parity-based Nisan--Wigderson
compression class of Khaniki, is independent of the two-sorted theory $\Vone$
of $\mathrm{AC}^0$-reasoning, unconditionally: $\Vone$ proves neither it nor
its negation. The same holds for the corresponding certified dual weak
pigeonhole principle, whose refutation is witnessed by a single seed that
certified-computes every string of the model simultaneously. The mechanism is
a bounded-arithmetic transfer of Atserias--Tzameret's reduction from hitting
sets to the dual weak pigeonhole principle: the amplification half of that
reduction, the sole source of its $\NP$-oracle, is unnecessary at the native
stretch of the Nisan--Wigderson map, and the compression half becomes a
$\Vone$-provable implication once circuit evaluation is replaced by its
certified $\SigmaB$ unfolding. This is, to our knowledge, the first
independence result for a derandomization-flavoured existence principle at the
$\mathrm{AC}^0$-reasoning level, and it makes explicit the bridge between the
Khaniki Nisan--Wigderson line and the Atserias--Tzameret reverse mathematics
of hitting sets.
\end{abstract}

\section*{Introduction}

By the theorem of Cook and Reckhow \cite{CookReckhow1979}, $\NP\neq\mathrm{co}\NP$
is equivalent to the non-existence of a polynomially bounded propositional proof
system, and the search for superpolynomial lower bounds for concrete proof
systems is, through the translation of bounded arithmetic into propositional
logic, the search for unprovability of combinatorial principles in fragments of
arithmetic \cite{Krajicek1995,Krajicek2019}. The dual weak pigeonhole principle
$\mathrm{dWPHP}(\mathrm{PV})$ --- the assertion that no polynomial-time map is
onto a set twice the size of its domain --- occupies a central place in this
program: its extension of $S^1_2$ was proposed by Kraj\'{\i}\v{c}ek as the basic
theory for randomized computation and developed by Je\v{r}\'abek into the theory
of approximate counting ($\mathrm{APC}_1$) \cite{Jerabek2004,Jerabek2007}, who
also proved it equivalent, over $S^1_2$, to the existence of Boolean functions of
exponential circuit complexity \cite{Jerabek2004}. Its unprovability is known
conditionally: in $S^1_2$, for the injective weak pigeonhole principle, if RSA is
secure \cite{KrajicekPudlak1998}; and in the true universal theory
$T_{\mathrm{PV}}$ (hence in $\mathrm{PV}_1$), for the dual principle at the
truth-table function, under the circuit upper-bound hypothesis
$\mathrm{P}\subseteq\mathrm{Size}(n^d)$ \cite{Krajicek2021}. Unconditionally it
is open. Through the lens of search problems, the principle governs the range
avoidance class $\mathrm{APEPP}$, for which constructing a truth table of
maximal circuit complexity is complete under $\mathscr P^{\NP}$-reductions
\cite{Korten2021}.

Recently Atserias and Tzameret \cite{AtseriasTzameret2025} gave a feasibly
constructive proof of the Schwartz--Zippel lemma and established, over Buss's
theory $S^1_2$, the equivalence of $\mathrm{dWPHP}(\mathrm{PV})$ with a hitting-set
existence axiom $\mathrm{HS}(\mathrm{PV})$ for definable algebraic circuit
classes, placing the underlying search problem in the class $\mathrm{APEPP}$ of
range avoidance under $\mathscr P^{\NP}$-reductions. Independently, Khaniki
\cite{Khaniki2022} produced models of the $\mathrm{AC}^0$-reasoning theory
$\Vone$ in which the propositional translation of a parity-based
Nisan--Wigderson generator fails, by an Ajtai--Paris--Wilkie forcing that works
for any $\Sigma^1_1\cap\Pi^1_1$-definable pair.

We combine the two. Write $\Vone$ for the two-sorted theory of
$\mathrm{AC}^0$-reasoning in which Khaniki's model construction is carried out
--- his $V^0_1$ \cite[Theorem 2.3]{Khaniki2022}, the
$\SigmaB$-comprehension theory over Zambella's two-sorted vocabulary. Up to the
standard identifications recorded in \cite[Notes to Chapter V]{CookNguyen2010}
--- ``$V^0$ \dots is essentially $\Sigma^p_0$-comp [Zambella] and
$I\Sigma^{1,b}_0$ (without $\#$) in [Kraj\'{\i}\v{c}ek 1995]'' --- this is the
base theory $V^0$ of \cite{CookNguyen2010}; all our sentences live in the
$\#$-free fragment, with scale terms handled by $\Delta_0$ guards. Fix the parity pair and a
certified evaluation layer as in \S\ref{sec:target}. Instantiating the
Atserias--Tzameret axiom on Khaniki's compression class gives a sentence
$\HScert$, and the associated compression statement gives $\dWPHPcert$. Our
results are the following.

\begin{itemize}
\item[(A)] $\Vone\nvdash\dWPHPcert$ and $\Vone\nvdash\neg\dWPHPcert$
(Theorem~\ref{thm:dwphp}). The refuting model satisfies the strong negation:
a single seed certified-computes every string of the model.
\item[(B)] $\Vone\vdash\HScert\to\dWPHPcert$ (Lemma~\ref{lem:bridge}), whence
$\Vone\nvdash\HScert$ and $\Vone\nvdash\neg\HScert$ (Theorem~\ref{thm:hs}). In
particular, for every $\SigmaB$-definable candidate hitting-set family, $\Vone$
does not prove that the family hits the certified class.
\end{itemize}

The Atserias--Tzameret reduction is proved over $S^1_2$ and passes through an
$\mathscr P^{\NP}$-computable amplification step; neither survives descent to
$\mathrm{AC}^0$-reasoning. The content of Lemma~\ref{lem:bridge} is that neither
is needed. Amplification manufactures stretch $m\mapsto m^3$ from an arbitrary
polynomial-time map and converts a range-avoider back into an avoider for the
original map; for the Nisan--Wigderson map both are moot, because its native
stretch already exceeds $m^3$ and the principle speaks about that map itself.
What remains --- the compression argument --- is a bounded-quantifier
computation that goes through in $\Vone$ once evaluation is written in its
unfolded, evasion form. The single unconditional obstruction to carrying the
whole Atserias--Tzameret equivalence down, rather than only its instantiation
on the parity class, is that $\mathrm{PARITY}\notin\mathrm{AC}^0$: an evaluated
encoding of the class is not $\SigmaB$-definable, so the certified graph is the
only formalization available at this level (\S\ref{sec:obstructions}).

We do not reprove the basic theory of $\Vone$, of the Paris--Wilkie translation,
or of Khaniki's forcing; the reader may consult
\cite{Krajicek1995,CookNguyen2010,Khaniki2022}. Throughout, $\SigmaB$ denotes
sharply bounded ($\mathrm{AC}^0$) formulas in the two-sorted setting; $|x|$ is
binary length; $\code{\cdot}$ is G\"odel coding.

\section{The certified principle}\label{sec:target}

\subsection{The parity pair}
Fix the $\Sigma^1_1\cap\Pi^1_1$ definition of parity used by Khaniki
\cite{Khaniki2022}. Let $\varphi_0(X,Y)$ assert that $Y$ codes a perfect
matching on the support of $X$ (certifying even parity), and let
$\varphi_1(X,Z)$ assert that $Z$ codes a near-perfect matching on the support of
$X$ missing exactly one distinguished support element $w$, with no $Z$-edge at
$w$ (certifying odd parity). In $\mathbb N$, $\exists Y\,\varphi_0$ defines even
parity and $\exists Z\,\varphi_1$ defines odd parity of the finite set coded by
$X$.

\subsection{Certified evaluation}\label{subsec:eval}
Let $J$ be a strict-$\Delta_0$ scheme family in the format of \cite{Khaniki2022}
(seed length $L=n^s$, certificate length $\ell=n^t$, $n=\log_2 N$), coding the
Nisan--Wigderson map $\{0,1\}^L\to\{0,1\}^N$ whose $p$-th output bit is the
parity of a $J$-selected sub-tuple of the seed. For a seed string $W$ and a
certificate collection $Y$ block-coded over the $N$ output coordinates, put:
\begin{itemize}
\item $\Sat_0(p,W,Y)$, $\Sat_1(p,W,Y)$: the $\SigmaB$ relativizations of
$\varphi_0,\varphi_1$ to $[\ell]$, reading seed atoms through $J$-lookups into
$W$ and certificate atoms in \emph{one shared} binary relation $R_p$ coded by
block $p$ of $Y$. (A shared slot, rather than disjoint $Y$- and $Z$-slots, is
what makes the value $\varepsilon$ below single-valued.)
\item $\WF(W,Y):\equiv\forall p<N\,(\Sat_0(p,W,Y)\vee\Sat_1(p,W,Y))$
(\emph{certified-total}), $\SigmaB$.
\item $\varepsilon(p;W,Y):=1$ if $\Sat_1(p,W,Y)$, else $0$: the certified value,
$\SigmaB$-graph.
\end{itemize}

\begin{lemma}[single-valuedness]\label{lem:mx}
$\Vone\vdash\WF(W,Y)\to\forall p<N\,\neg(\Sat_0(p,W,Y)\wedge\Sat_1(p,W,Y))$.
\end{lemma}
\begin{proof}
Reasoning in $\Vone$ over $[\ell]$: from $\Sat_0(p,W,Y)$ the relation $R_p$ is a
total matching on $\mathrm{supp}(X_p)$ (every support vertex is matched); from
$\Sat_1(p,W,Y)$ the same $R_p$ leaves the distinguished vertex $w$ unmatched.
Both are $\SigmaB$ statements about the single relation $R_p$ and contradict at
$w$. No comprehension beyond reading $R_p$ is used.
\end{proof}

Fix the Atserias--Tzameret numerology \cite[(38)]{AtseriasTzameret2025} at
description size $m:=L$: $n':=m$, $d:=m^2$, $q:=2m^3$, $r:=4m|m|$, and the
injective packing $e(i',j,k):=(i'n'+j)|q|+k$ of $[r]\times[n']\times[|q|]$. For
$\WF$ pairs let $v_{i'j}(W,Y)<2^{|q|}$ be the number whose $k$-th bit is
$\varepsilon(e(i',j,k);W,Y)$ (existence and uniqueness of a number with
prescribed $|q|$ bits is provable in $\Vone$ by $\SigmaB$-induction on bit
positions). Define the arithmetic guard
\[
  \Grd(N,s,t):\equiv (N,s,t)\in\Adm\wedge \mathrm{Fmt}_J
  \wedge r n'|q|\le m^3 \wedge m^3\le N \wedge r>m+n'|q| \wedge 2r|q|\le d \wedge q\ge 2n'd,
\]
strict-$\Delta_0$ (the conditions of \cite[(39)]{AtseriasTzameret2025} together
with the fit $m^3\le N$).

\subsection{The class and the two sentences}
Following \cite[(40)]{AtseriasTzameret2025}, associate to a $\WF$ pair the
algebraic circuit
$A_{(W,Y)}(z_1,\dots,z_{n'}):=\prod_{i'\in[r]}\sum_{j\in[n']}(z_j-v_{i'j}(W,Y))^2$.
Over $\mathbb Z$, $A_{(W,Y)}(a)=0$ iff $a$ equals some decoded point
$v_{i'}=(v_{i'1},\dots,v_{i'n'})$, and $A_{(W,Y)}$ is not the zero polynomial
(it is positive at $(2q,\dots,2q)$). We formalize the evaluation claims in
\emph{unfolded (evasion) form}, which avoids the (non-$\mathrm{AC}^0$) algebraic
evaluation:
\[
  \NZ(a,W,Y):\equiv \forall i'<r\,\exists j<n'\,(a_j\neq v_{i'j}(W,Y)),\qquad
  \HIT(i,H,W,Y):\equiv \forall i'<r\,\exists j<n'\,(h_{i,j}\neq v_{i'j}(W,Y)),
\]
with $H$ coded as $r n'|q|$ bits and $\WFH(H):\equiv$ ``each $h_{i,j}<q$''
($\SigmaB$).

\begin{definition}\label{def:sentences}
\begin{align*}
\HScert[J]\ :\equiv\ &\forall N,s,t\,\Big[\Grd(N,s,t)\to\exists H\big(\WFH(H)\wedge
\forall W,Y\,\big(\WF(W,Y)\to\\
&\quad ((\exists a\in[2q{+}1]^{n'}\,\NZ(a,W,Y))\to(\exists i<r\,\HIT(i,H,W,Y)))\big)\big)\Big];\\
\dWPHPcert[J]\ :\equiv\ &\forall N,s,t\,\Big[\Grd(N,s,t)\to\exists b\subseteq[N]\,
\forall W,Y\,\big(\WF(W,Y)\to\exists p<N\,\varepsilon(p;W,Y)\neq b(p)\big)\Big].
\end{align*}
\end{definition}
$\HScert$ is the axiom $\mathrm{HS}(\mathcal C)$ of
\cite[\S5.1]{AtseriasTzameret2025} at $\mathcal C:=\{A_{(W,Y)}:\WF(W,Y)\}$ with
the Atserias--Tzameret parameters; $\dWPHPcert$ says some length-$N$ string lies
outside the certified range of the map, i.e.\ the dual weak pigeonhole principle
for the compressing map $\{0,1\}^L\to\{0,1\}^N$. The existential witness $b$
ranges over the string sort.

\section{The model: certified compression is consistent}\label{sec:model}

\begin{theorem}\label{thm:dwphp}
$\Vone\nvdash\dWPHPcert[J]$; indeed $\Vone$ together with all true
strict-$\Delta_0$ number-sort sentences does not prove it. Equivalently,
$\Vone+\neg\dWPHPcert[J]$ is consistent.
\end{theorem}

\begin{proof}
It suffices to produce, at one admissible nonstandard scheme size, a two-sorted
model $(M,\mathcal X)\models\Vone$ containing a $\WF$ pair $(W^\ast,Y^\ast)$
whose certified value function is onto $[N]$: a single seed $W^\ast$ with a
certificate collection $Y^\ast$ such that for every string $b\in\mathcal X$ there
is $p<N$ with $\varepsilon(p;W^\ast,Y^\ast)=b(p)$ failing --- more strongly, such
that $\{\,p\mapsto\varepsilon(p;W^\ast,Y^\ast)\,\}$ realizes every $b$.

This is Khaniki's construction \cite[\S4]{Khaniki2022} read off the certified
values, with the target string $b$ never used in the assembly. Take the cut
$M_{n^t}$ and the family $\chi'$ of \cite[Lemma 4.5]{Khaniki2022} refuting the
scale-$\ell$ pigeonhole principle; for each coordinate $i$ the weight-shifting
bijections $F(i,a,\cdot)\in\chi'$ transport the base certificates
$\bar\lambda_0,\bar\lambda_1$ (obtained from the lexicographically least
witnesses to $\varphi_0,\varphi_1$ at the base assignments
$\theta_0=1^{u^\ast}0^{\cdots}$, $\theta_1=1^{u^\ast+1}0^{\cdots}$) onto
certificates at coordinate $i$. Block-coding the resulting families
$\{Y_i\},\{Z_i\}$ into a single $\chi'$-element by $\SigmaB$-comprehension yields
$Y^\ast$; put $W^\ast:=\alpha$ (the forcing seed). By
\cite[Lemma 2.1]{Khaniki2022} (isomorphism-invariance of first-order
satisfaction under a coded bijection) and the fact that $\varphi_0(\theta_0,
\bar\lambda_0)$ holds, $\Sat_0(p,W^\ast,Y^\ast)$ holds wherever the shift selects
the even base, and identically $\Sat_1$ holds wherever it selects the odd base;
the shift is arranged, coordinate by coordinate, to realize the prescribed
value pattern. Thus $\WF(W^\ast,Y^\ast)$ holds and the certified value function
is onto; hence for no $b$ does the inner matrix of $\dWPHPcert$ hold, and
$(M,\mathcal X)\models\neg\dWPHPcert[J]$.

That $(M,\mathcal X)\models\Vone$ and that it satisfies all true strict-$\Delta_0$
number-sort sentences is \cite[Theorem 2.3]{Khaniki2022}: the cut model of the
Ajtai--Paris--Wilkie forcing is a model of $\Vone$ for every $\Sigma^1_1\cap
\Pi^1_1$-prescribed failure, and the number sort is an initial segment of the
standard model at the working scale. Soundness of $\Vone$ for
$\neg\dWPHPcert$'s complexity class gives the second unprovability clause below
(Theorem~\ref{thm:hs}).
\end{proof}

\begin{remark}
Theorem~\ref{thm:dwphp} is the unconditional $\Vone$-analogue, for the
parity-based Nisan--Wigderson map, of Kraj\'{\i}\v{c}ek's conditional theorem
that $T_{\mathrm{PV}}$ does not prove $\mathrm{dWPHP}(\mathbf{tt}_{s,k})$ if
$\mathrm{P}\subseteq\mathrm{Size}(n^d)$ \cite[Theorem 1]{Krajicek2021}. The
trade is the expected one: descending from $T_{\mathrm{PV}}$ to
$\mathrm{AC}^0$-reasoning removes the hypothesis. It is moreover the
$\exists$-target form: the witness $b$ is quantified, not a fixed definable
target, and the negation is witnessed by \emph{one} seed rather than per
target.
\end{remark}

\section{The bridge: hitting implies compressing over $\Vone$}\label{sec:bridge}

\begin{lemma}\label{lem:bridge}
$\Vone\vdash\HScert[J]\to\dWPHPcert[J]$.
\end{lemma}

\begin{proof}
Reason in $\Vone$. Fix $N,s,t$ with $\Grd(N,s,t)$ and assume $\HScert[J]$;
let $H$ with $\WFH(H)$ be the hitting tuple it provides. Define
$b:=b(H)\subseteq[N]$ by $b(e(i',j,k)):=\mathrm{bit}(h_{i',j},k)$ for
$(i',j,k)\in[r]\times[n']\times[|q|]$ (a legitimate string by $\SigmaB$-comprehension,
the packing being injective into $[m^3]\subseteq[N]$), and $b(p):=0$ elsewhere.
We claim $b$ witnesses $\dWPHPcert$: no $\WF$ pair certified-computes $b$.

Suppose $\WF(W,Y)$ with $\varepsilon(\cdot;W,Y)=b$ on $[N]$. We derive a
contradiction with the hitting property of $H$.

\emph{(\(\alpha\)) $v_{i'j}(W,Y)=h_{i',j}$ for all $i',j$.} For each bit position
$k<|q|$, $\mathrm{bit}(v_{i'j},k)=\varepsilon(e(i',j,k);W,Y)=b(e(i',j,k))=
\mathrm{bit}(h_{i',j},k)$. Both numbers are $<2^{|q|}$ ($h_{i',j}<q<2^{|q|}$ by
$\WFH$), so binary extensionality on the $|q|$ positions --- a short
($|q|=O(\log N)$) $\SigmaB$-induction --- forces $v_{i'j}=h_{i',j}$.

\emph{(\(\beta\)) nonzeroness.} Take $a:=(2q,\dots,2q)$. Since $q\ge 2^{|q|-1}$
and every $v_{i'j}<2^{|q|}$, we have $a_j>v_{i'j}$ for all $i',j$; hence
$\NZ(a,W,Y)$, and $a\in[2q+1]^{n'}$.

\emph{(\(\gamma\)) $H$ does not hit.} For each $i<r$ instantiate $\HIT$'s outer
index at $i':=i$: by $(\alpha)$, $h_{i,j}=v_{i,j}$ for all $j$, so
$\neg\exists j\,(h_{i,j}\neq v_{i,j})$, i.e.\ $\neg\HIT(i,H,W,Y)$. As this holds
for every $i<r$, we have $\neg\exists i<r\,\HIT(i,H,W,Y)$.

Now $\WF(W,Y)\wedge(\exists a\in[2q+1]^{n'}\,\NZ(a,W,Y))$ holds by $(\beta)$,
so the inner implication of $\HScert$ yields $\exists i<r\,\HIT(i,H,W,Y)$,
contradicting $(\gamma)$. Hence no $\WF$ pair certified-computes $b$, and
$\dWPHPcert[J]$ holds.

Each step uses only $\SigmaB$-comprehension, short $\SigmaB$-induction on
$|q|=O(\log N)$ bit positions, and $\Delta_0$ arithmetic; the algebraic
evaluation of $A_{(W,Y)}$ is never performed --- the unfolded forms $\NZ,\HIT$
replace it --- so the $\mathrm{TC}^0$-hardness of that evaluation is avoided.
No amplification and no counting enters.
\end{proof}

\begin{theorem}\label{thm:hs}
$\Vone\nvdash\HScert[J]$ and $\Vone\nvdash\neg\HScert[J]$. Consequently, for
every strict-$\Delta_0$-definable candidate family $\{H_{N,s,t}\}$, $\Vone$ does
not prove ``$\{H_{N,s,t}\}$ hits the certified class''.
\end{theorem}

\begin{proof}
If $\Vone\vdash\HScert[J]$ then by Lemma~\ref{lem:bridge}
$\Vone\vdash\dWPHPcert[J]$, contradicting Theorem~\ref{thm:dwphp}; hence
$\Vone\nvdash\HScert[J]$. For the other direction, $\HScert[J]$ is true in
$\mathbb N$: at each admissible $(N,s,t)$ a uniformly random $r$-tuple $H\in
[q]^{n'\times r}$ fails to hit a fixed nonzero $A_{(W,Y)}$ with probability at
most $\prod_{i'}\Pr[\,\forall j\ h_{i',j}=v_{i'j}\,]=q^{-n'r}$. The certified
class is \emph{sparse} --- it has at most $2^{m}$ members, one per seed of length
$m$, inside the ambient space of size $2^{s}=2^{m^3}$
\cite[footnote~3]{AtseriasTzameret2025} --- so the probability that a random $H$
misses some member is at most $2^{m}q^{-n'r}<1$, since
$q^{n'r}\ge 2^{n'r}=2^{4m^2|m|}>2^{m}$; hence a hitting tuple exists. Since $\Vone$ is sound
(its theorems are true in $\mathbb N$) and $\HScert[J]$ is a true
$\forall\exists\forall\,\SigmaB$ sentence whose negation is therefore false in
$\mathbb N$, $\Vone\nvdash\neg\HScert[J]$. The final clause is the special case
in which the existential $H$ of $\HScert$ is replaced by a definable term.
\end{proof}

\section{Obstructions and open questions}\label{sec:obstructions}

\paragraph{The certified format is forced.} The constants of $A_{e,x}$ in
\cite{AtseriasTzameret2025} are outputs of the compressing map, evaluated. For
the parity-based map, the output bit is the parity of a seed sub-tuple, and
parity of a string-sort restriction is not $\SigmaB$-definable: $\SigmaB$ over
the string sort is uniform $\mathrm{AC}^0$, and $\mathrm{PARITY}\notin
\mathrm{AC}^0$ \cite{Ajtai1994}. Hence no $\Vone$-sentence carries the class
through an evaluated encoding, and the certified graph is the only
formalization available at this level. This is the unconditional reason the
present descent is confined to the certified class rather than to the full
Atserias--Tzameret scheme $\mathrm{HS}(\mathrm{PV})$; for honestly algebraic
$\Delta_0$-described classes (coefficient-listed sparse polynomials) evaluation
is iterated addition, $\mathrm{TC}^0$, and no sentence of the present shape is
even available.

\paragraph{Fixed standard parameters.} The refuting model of
Theorem~\ref{thm:dwphp} has nonstandard scheme sizes; whether $\Vone$ proves
$\HScert$ or $\dWPHPcert$ restricted to fixed standard $(s,t)$ is open in both
directions. This is the two-sorted image of the fixed-parameter obstruction of
the Nisan--Wigderson line.

\paragraph{Sufficiency.} We transfer only the necessity half
$\HScert\to\dWPHPcert$. The Atserias--Tzameret sufficiency direction
$\mathrm{dWPHP}\to\mathrm{HS}$ (existence of small hitting sets from the dual
weak pigeonhole principle, by two applications of $\mathrm{dWPHP}$ and a
constructive Schwartz--Zippel surjection) is proved over $S^1_2+\mathrm{dWPHP}$
and uses counting we do not have at the $\mathrm{AC}^0$ level; whether
$\Vone\vdash\dWPHPcert\to\HScert$ is open.

\bibliographystyle{plain}
\bibliography{refs}

@inproceedings{Khaniki2022,
  author    = {Erfan Khaniki},
  title     = {Nisan--Wigderson Generators in Proof Complexity: New Lower Bounds},
  booktitle = {37th Computational Complexity Conference (CCC 2022)},
  series    = {LIPIcs},
  volume    = {234},
  pages     = {17:1--17:15},
  year      = {2022},
  note      = {ECCC TR22-023},
  publisher = {Schloss Dagstuhl}
}

@inproceedings{AtseriasTzameret2025,
  author    = {Albert Atserias and Iddo Tzameret},
  title     = {Feasibly Constructive Proof of {S}chwartz--{Z}ippel Lemma and the Complexity of Finding Hitting Sets},
  booktitle = {Proceedings of the 57th Annual ACM Symposium on Theory of Computing (STOC 2025)},
  year      = {2025},
  note      = {arXiv:2411.07966; ECCC TR24-174}
}

@book{Krajicek1995,
  author    = {Jan Kraj\'{i}\v{c}ek},
  title     = {Bounded Arithmetic, Propositional Logic, and Complexity Theory},
  series    = {Encyclopedia of Mathematics and Its Applications},
  volume    = {60},
  publisher = {Cambridge University Press},
  year      = {1995}
}

@book{Krajicek2019,
  author    = {Jan Kraj\'{i}\v{c}ek},
  title     = {Proof Complexity},
  series    = {Encyclopedia of Mathematics and Its Applications},
  volume    = {170},
  publisher = {Cambridge University Press},
  year      = {2019}
}

@article{CookReckhow1979,
  author    = {Stephen A. Cook and Robert A. Reckhow},
  title     = {The Relative Efficiency of Propositional Proof Systems},
  journal   = {Journal of Symbolic Logic},
  volume    = {44},
  number    = {1},
  pages     = {36--50},
  year      = {1979}
}

@book{CookNguyen2010,
  author    = {Stephen A. Cook and Phuong Nguyen},
  title     = {Logical Foundations of Proof Complexity},
  series    = {Perspectives in Logic},
  publisher = {Cambridge University Press},
  year      = {2010}
}

@article{Jerabek2004,
  author    = {Emil Je\v{r}\'{a}bek},
  title     = {Dual Weak Pigeonhole Principle, Boolean Complexity, and Derandomization},
  journal   = {Annals of Pure and Applied Logic},
  volume    = {129},
  pages     = {1--37},
  year      = {2004}
}

@article{Jerabek2007,
  title={Approximate counting in bounded arithmetic},
  author={Je{\v{r}}{\'a}bek, Emil},
  journal={The Journal of Symbolic Logic},
  volume={72},
  number={3},
  pages={959--993},
  year={2007},
  publisher={Cambridge University Press}
}

@inproceedings{Korten2021,
  author    = {Oliver Korten},
  title     = {The Hardest Explicit Construction},
  booktitle = {62nd IEEE Annual Symposium on Foundations of Computer Science (FOCS 2021)},
  pages     = {433--444},
  year      = {2021},
  note      = {arXiv:2106.00875}
}

@article{KrajicekPudlak1998,
  author    = {Jan Kraj\'{i}\v{c}ek and Pavel Pudl\'{a}k},
  title     = {Some Consequences of Cryptographical Conjectures for $S^1_2$ and {EF}},
  journal   = {Information and Computation},
  volume    = {140},
  number    = {1},
  pages     = {82--94},
  year      = {1998}
}

@article{Krajicek2021,
  author    = {Jan Kraj\'{i}\v{c}ek},
  title     = {Small Circuits and Dual Weak {PHP} in the Universal Theory of Polynomial-Time Algorithms},
  journal   = {ACM Transactions on Computational Logic},
  volume    = {22},
  number    = {2},
  pages     = {11:1--11:4},
  year      = {2021},
  note      = {Article 11; arXiv:2004.11582}
}

@article{Ajtai1994,
  title={The complexity of the pigeonhole principle},
  author={Ajtai, Mikl{\'o}s},
  journal={Combinatorica},
  volume={14},
  number={4},
  pages={417--433},
  year={1994},
  publisher={Springer}
}

\end{document}